\documentclass[a4paper,showpacs,superscriptaddress,twocolumn,aps,pra]{revtex4-1}

\usepackage{amsmath}
\usepackage{amsfonts}
\usepackage{amssymb}
\usepackage{amscd}
\usepackage{graphics}
\usepackage{graphicx}
\usepackage{psfrag}
\usepackage{color}
\usepackage{colordvi}
\usepackage{bbm}

\usepackage{amsthm}

\usepackage{bbm}
\newcommand{\unity}{\mathbbmss{1}}

\newcommand{\opvec}{\operatorname{vec}}

\newcommand{\R}{\ensuremath{\mathbb R}}
\newcommand{\Q}{\ensuremath{\mathbb Q}}
\newcommand{\C}{\ensuremath{\mathbb C}}
\newcommand{\Z}{\mathbb{Z}}

\newcommand{\bra}[1]{\ensuremath{\langle #1 |}{}}
\newcommand{\ket}[1]{\ensuremath{| #1 \rangle}{}}

\newcommand{\lie}[1]{\langle #1 \rangle}
\newcommand{\Cd}{\ensuremath{{\mathbb C}^{d\times d}}}
\newcommand{\Cdd}{\ensuremath{{\mathbb C}^{d^2\times d^2}}}

\newcommand{\cent}{\mathcal{C}}

\newcommand{\rank}{\mathrm{rank}}

\newcommand{\uu}{\mathfrak{u}}
\newcommand{\su}{\mathfrak{su}}
\newcommand{\SU}{\operatorname{SU}}

\newcommand{\fc}{\mathfrak{c}}

\newcommand{\fg}{\mathfrak{g}}
\newcommand{\fh}{\mathfrak{h}}

\newcommand{\fs}{\mathfrak{s}}

\newcommand{\tr}{\mathrm{Tr}}
\newcommand{\PP}{\ensuremath{\mathcal{P}}}
\newcommand{\QQ}{\ensuremath{\mathcal{Q}}}
\newcommand{\ts}[1]{{#1}^{(2)}}

\makeatletter 
  \renewcommand{\@upn}[1]{#1} 
\makeatother 

\makeatletter
\newtheoremstyle{ex}%
  {0pt}%
  {0pt}%
  {}%
  {\parindent}%
  {\itshape}%
  {:}
  {.5em}
  {\thmname{#1}\thmnumber{\@ifnotempty{#1}{ }\@upn{#2}}%
  \thmnote{ {\the\thm@headfont(#3)}}}%
\makeatother

\theoremstyle{ex}  
\makeatletter
\newtheorem*{rep@theorem}{\rep@title}
\newcommand{\newreptheorem}[2]{%
\newenvironment{rep#1}[1]{%
 \def\rep@title{#2 {\@upn{\ref{##1}}}}%
 \begin{rep@theorem}}%
 {\end{rep@theorem}}}
\makeatother

\newtheorem{result}{Result}
\newtheorem*{main*}{Main result}

\newtheorem{example}{Example}
\newreptheorem{example}{Example}

\newtheorem{theorem}{Theorem}

\newtheorem{proposition}[theorem]{Proposition}

\newtheorem*{conjecture*}{Conjecture}

\usepackage[bookmarks=false,pdfstartview={FitH}]{hyperref}

\begin{document}

\title{Symmetry criteria for quantum simulability of effective interactions}

\author{Zolt{\'a}n Zimbor{\'a}s}
\email{zimboras@gmail.com}
\affiliation{Department of Computer Science, University College London,
Gower Street, London WC1E 6BT, United Kingdom}

\author{Robert Zeier}
\email{zeier@ch.tum.de}
\affiliation{Department Chemie, Technische Universit{\"a}t M{\"u}nchen,
Lichtenbergstrasse 4, 85747 Garching, Germany}

\author{Thomas \surname{Schulte-Herbr{\"u}ggen}}
\email{tosh@ch.tum.de}
\affiliation{Department Chemie, Technische Universit{\"a}t M{\"u}nchen,
Lichtenbergstrasse 4, 85747 Garching, Germany}

\author{Daniel Burgarth}
\email{dkb3@aber.ac.uk}
\affiliation{Department of Mathematics, Aberystwyth University, Aberystwyth SY23 2BZ, United Kingdom}

\date{October 13, 2015}


\begin{abstract}
What can one do with a given tunable quantum device? 
We provide complete symmetry criteria deciding whether some effective target interaction(s) can 
be simulated by a set of given interactions.
Symmetries lead to a better understanding of simulation
and permit a reasoning beyond the limitations of the usual explicit Lie closure.
Conserved quantities induced by symmetries pave the way to
a resource theory for simulability.
On a general level, one can now
decide equality for any pair of compact  Lie algebras just given by their generators
without determining the algebras explicitly. 
Several physical examples are illustrated,  including entanglement invariants, 
the relation to unitary gate membership problems, 
as well as the central-spin model. 
\end{abstract}

\maketitle

\section{Introduction}
Thanks to impressive progress on the experimental side, many small- and medium-scale quantum devices are now 
ready for applications ranging from quantum metrology 
\cite{Huelga98,BDM99,ASL04,GLM11}
to quantum simulation 
\cite{Lloyd96,BCL+02,Wunder09,LWG10,CMLS12}.
With quantum information processing as one of the driving but long-term goals (e.g., 
\cite{Lloyd08,RDN12,CNR+14b}),
one of the pressing questions is, what can one do with these devices \emph{now}? 
This problem clearly falls into the remit of \emph{quantum systems and control engineering}, an area naturally 
receiving increased interest 
\cite{DowMil03,dAll08,WisMil09}
both experimentally and theoretically.

Control theory offers a well-known characterization of the operations that a quantum device is capable of on
\emph{Lie-algebraic} grounds~\cite{JS72,Jurdjevic97,DiHeGAMM08,dAll08,ZS11,KDH12,ZZK14}. 
In this work,
we simplify the question to the {\em Hamiltonian membership problem}
of (finite-dimensional) quantum simulation. 
It amounts to deciding, for a set of {\em given control interactions} $\PP$, whether a set of 
\emph{effective target interactions} $\QQ$ can be simulated---without having to establish controllability 
via nested (and hence tedious) commutator calculations for the so-called Lie closure.
Our results reduce the Hamiltonian membership problem to the straightforward solution
of homogeneous linear equations.

In the setting of the  controlled Schr{\"o}dinger equation~\cite{Sak:1994} 
(taken as a bilinear control system \cite{Elliott09,Jurdjevic97}) 
\begin{equation}\label{control_system}
\frac{d}{dt}U(t) = [-iH_1 + {\sum}_{\nu=2}^{p} -i u_\nu(t) H_\nu  ]\, U(t),
\end{equation}
we ask whether
the given set $\PP:=\{iH_1,\ldots,iH_p\}$ of interactions (which may include a drift term) generates
an effective interaction $iH_{p+1}$ or more generally any interaction from a set $\QQ:=
\{iH_{p+1},\ldots,iH_q\}$ assuming all $H_\nu$ are represented by Hermitian matrices henceforth. 
If so, then for \emph{every} evolution time $\tau>0$
of a {\em simulated interaction $iH_{k}\in \QQ$},
there is a solution $U(t)$ of the {\em simulating system} \eqref{control_system}
for $0\leq t \leq  \theta$ and controls $u_\nu(t)$
such that $\PP$ generates a unitary $U(\theta)=\exp(-i\tau H_{k})$
in the simulation time $\theta$
starting from the identity at $t=0$
\cite{HW:1968,Khaneja01b,WJB02,BCL+02,
VHC02,BDD02,ZGB,BBN05},
\footnote{The simulation time $\theta$ can be infinite in order to also cover peculiar 
cases such as an irrational winding of a torus. 
Note that the Hamiltonian evolution is \emph{not} necessarily simulated continuously 
during a time interval, but we only assume that the correct total evolution 
$\exp(-i\tau H_k)$
is attained after a suitably chosen duration $\theta$ which may depend
on the arbitrary but \emph{fixed} evolution time $\tau$.}.
In this sense, Hamiltonian simulation of a particular Hamiltonian $H_k$ can be considered as an
infinitesimal version of creating a particular unitary gate.
It also
generalizes the
universality (or full controllability) question of whether \emph{all} Hamiltonians
can be simulated
(or equivalently whether \emph{all} unitary gates can be obtained)
\cite{Barenco,RSD+95,BB:2002,BDD02,SchiFuSol01,TR01,SchiPuSol01,Alt02,EGK96,
AA03,PST09,ZS11}. 
In the context of gates, a familiar elementary example is that 
all unitary gates in an $n$-qubit system can be obtained \cite{Barenco}
by combining  local gates  with {\sc cnot} gates.
However, the approach of the pioneering 
age of decomposing every target gate into a sequence of {\sc cnot} and local gates is, in practice,
all too often imprecise or slow. So implementing gates or simulating Hamiltonians
with high fidelity rather asks for optimal control techniques, as explained in a recent
roadmap~\cite{qcontrol-roadmap2015}.
As a precondition, here we step back to the Hamiltonian level and give criteria for simulability
and controllability.

\section{Main idea}

We solve the decision problems of simulability (and controllability) by just analyzing the
symmetries of the Hamiltonians of given setups. We show that this decision requires 
considering both {\em linear and quadratic symmetries}, where linear symmetries of a Hamiltonian
$H$ commute with $H$, while {\em quadratic symmetries} of $H$ are those 
commuting with the tensor square $(iH\otimes\unity+\unity\otimes iH)$.
The term quadratic symmetry is motivated, since the tensor square  
generates $U\otimes U$ just as $iH$ generates the unitary~$U$.

More precisely, our goal is to get a symmetry-based understanding of how a set $\PP$ 
of available interactions can simulate a set $\QQ$  of desired effective quantum interactions 
in the sense that the Lie closures
coincide, i.e. $\lie{\PP} = \lie{\PP{\cup}\QQ}$. We circumvent brute-force 
calculation of the Lie closure not only because
high-order commutators can entail
a significant growth in the appearing matrix entries and 
may lead to instabilities in numerical computations, but first and foremost because it
provides no deeper insight into the problem.
Our symmetry analysis leads to a much more systematic understanding 
of Hamiltonian simulation and quantum system dynamics in general.
It provides a powerful argument to decide under which conditions a desired Hamiltonian can, 
in fact, be simulated, or in turn, which explicit simulations or computations are impossible 
in a given experimental setup.

Let us summarize our line of thought:
As 
short-hand, let the \emph{linear symmetries} of $\PP$ (analogously for any set of matrices) 
be expressed via the {\em commutant} 
$\PP'$ which consists of all matrices $S\in\mathbb C^{d\times d}$ that
commute (i.e., $[S,iH_\nu]=0$) 
with each element $iH_\nu \in \Cd$ of $\PP$ \footnote{In this work, $\Cd$ denotes the set of complex 
$d\times d$ matrices and $\unity_d$ signifies the $d\times d$ identity matrix.}. 
Obviously, for $\QQ$ to be simulable by $\PP$, it is necessary 
that $\QQ$ may not break but rather has to inherit the symmetries of $\PP$, so 
$\dim[\PP'] = \dim[(\PP {\cup} \QQ)']$.
However, a complete symmetry characterization is nontrivial. It rather requires the following two steps:
The first is to introduce \emph{quadratic} symmetries~\cite{ZS11}
as those linear symmetries of the system
artificially doubled by the \emph{tensor square}
$\PP^{\otimes 2}:=\{iH_\nu{\otimes}\unity_d {+} \unity_d{\otimes} iH_\nu\;\text{for}\; \nu\in\{1,\ldots, p\}\}$. 
It defines the 
{\em quadratic symmetries} by its commutant $\ts{\PP}:=(\PP^{\otimes 2})'$.
Secondly, let $\cent$ denote the center~\footnote{The center 
of a set $\mathcal{M}$ of matrices contains all
$M_1\in\mathcal{M}$ that commute (i.e.,  $[M_1,M_2]=0$) with every
$M_2 \in \mathcal{M}$.}
of the commutant $(\PP {\cup} \QQ)'$ 
and consider the central projections 
of $\PP$ and $\PP {\cup} \QQ$
onto $\cent$.
With these stipulations, we summarize our main result:

\begin{main*}[see Result~\ref{res_one}  below]
The given interactions $\PP$ simulate the desired 
interactions $\QQ$ in the sense $\lie{\PP} = \lie{\PP {\cup} \QQ}$ if and only if 
$\PP$ and $\QQ$ share the same \emph{quadratic} symmetries 
(i.e., $\dim[\ts{\PP}] = \dim[\ts{(\PP {\cup} \QQ)}]$) [condition~(A)]
and the
central projections of $\PP$ and $\PP{\cup}\QQ$ onto $\cent$ are of the same rank
[condition~(B)].
\end{main*}
\newcounter{centfoot}
\setcounter{centfoot}{\value{footnote}}

Let us emphasize that our approach goes 
beyond the ubiquitous use of linear symmetries in physics, since
linear symmetries provide only an incomplete picture of
Hamiltonian simulation.
The application of higher symmetries is the key here.
It is interesting to note that essentially only the quadratic symmetries (and no higher ones)
in condition (A) are necessary to characterize the dynamics of a quantum system.
One obtains a complete description together with the auxiliary condition (B).

Some remarks also summarizing known approaches are in order.
The quadratic symmetries are stronger than the linear ones; actually they include them 
and thus condition (A) implies that the linear symmetries also agree.
Example~\ref{ex_one} below illustrates why
matching the linear symmetries does not suffice to ensure simulability.
As shown in a companion paper~\cite{ZZ15b}, one can decide if
a subalgebra 
$\fh\subseteq\fg$ of a compact semisimple Lie algebra $\fg$ actually fulfills $\fh=\fg$
(e.g., $\lie{\PP} = \lie{\PP{\cup}\QQ}$)
 just by analyzing quadratic symmetries.
But Example~\ref{ex_two} elucidates  why condition~(A) alone does not,
in the general compact case, imply simulability. 
Only after fixing the central projections by condition (B) the quadratic symmetries
decide simulability.

On a much more general scale, condition (B) closes the gap to completely
characterizing equality in $\fh\subseteq\fg$ now for
{\em all compact Lie algebras} (generated by skew-Hermitian interactions) beyond the 
semisimple ones of \cite{ZZ15b}. Simplifying 
within the Lie-algebraic frame,
our symmetry approach to decide simulability 
and the membership $\QQ \subseteq\lie{\PP}$
can thus be seen as a major step beyond the well-established Lie-algebra rank 
condition \cite{JS72,Jurdjevic97,dAll08}
and beyond the limited first use of
quadratic symmetries to establish 
full 
controllability in~\cite{ZS11}.

\section{Symmetries}

In this section, we elaborate our method and
establish necessary and sufficient conditions for Hamiltonian simulation 
to arrive at  Result~\ref{res_one} below.
We also describe important properties of linear and quadratic symmetries
and discuss two illustrating examples. Example~\ref{ex_one} highlights the importance
of quadratic symmetries for deciding Hamiltonian simulation and their relevance 
for entanglement invariants. The necessity for the auxiliary condition (B)
is made evident in Example~\ref{ex_two}.

The linear symmetries
 of $\mathcal{M}\subseteq \Cd$ are identified \cite{ZS11} with the commutant $\mathcal{M}'$  given as
\begin{equation*}\label{eq_comm}
\mathcal{M}':=\{S \in \Cd \;\text{s.t.}\;  [S,M]=0 \,\text{ for all }\, M\in\mathcal{M}\}.
\end{equation*}
The commutant includes all
complex multiples of the identity $\unity_d$ and it
forms a vector space of dimension $\dim(\mathcal{M}')$.
A smaller set of matrices typically shows \emph{more} symmetries, i.e.,
for $\mathcal{M}_1 \subseteq \mathcal{M}_2$ one has
$\mathcal{M}_1'\supseteq \mathcal{M}_2'$
and  $\mathcal{M}_1'= \mathcal{M}_2'$ iff $\dim(\mathcal{M}_1')=\dim(\mathcal{M}_2')$.
By Jacobi's identity 
(i.e., $[S,[M_1,M_2]]=[[M_2,S],M_1]+[[S,M_1],M_2]$),
any symmetry $S$ that commutes with both $M_1$ and $M_2$
also commutes with their commutator  $[M_1,M_2]$. 
So, $\mathcal{M}$ and the Lie algebra 
$\lie{\mathcal{M}}$  it generates have the same commutant:
$\mathcal{M}_1'=\mathcal{M}_2'$ if $\lie{\mathcal{M}_1}=\lie{\mathcal{M}_2}$.

In our context, this implies that $iH_{p+1}$ cannot be simulated by $\PP$ unless
$\PP' = (\PP\cup\{iH_{p+1}\})'$, i.e.,
coinciding symmetries are a necessary but not sufficient condition. 
This is because the converse does not hold as the following basic example illustrates:
\begin{example}\label{ex_one}
The pair interaction 
$iH_{\sf zz}:=i Z_1 Z_2$ cannot be simulated
by the local interactions  $\PP=\{i X_1,\allowbreak i Y_1,\allowbreak  i X_2,\allowbreak i Y_2\}$ 
of a two-qubit system \footnote{Here,  
$X_k$ 
denotes the matrix $X:=
\left(\begin{smallmatrix}
0 & 1\\
1 & 0
\end{smallmatrix}\right)$ occurring at position $k$ in
$\unity_2 \otimes \cdots \otimes \unity_2 \otimes X \otimes \unity_2 \otimes \cdots \otimes \unity_2$;
also for the other Pauli matrices
$Y :=
\left(\begin{smallmatrix}
0 & -i\\
i & 0
\end{smallmatrix}\right)$ and
$Z :=
\left(\begin{smallmatrix}
1 & 0\\
0 & -1
\end{smallmatrix}\right)$.} in spite of coinciding (trivial) commutants
$\PP'=(\PP\cup\{iH_{\sf zz}\})'=\C \unity_4$.
\end{example}

Thus, we further discuss \emph{quadratic symmetries} \cite{ZS11} 
defined by the commutant to the tensor-square 
\footnote{
	The interaction $iH{\otimes}\unity_d {+} \unity_d{\otimes}iH$
	generates the unitary
	$\exp(-itH){\otimes}\exp(-itH)=\exp[-t(iH{\otimes}\unity_d {+} \unity_d{\otimes}iH)]$.},
\begin{align*}
\ts{\mathcal{M}} 
&:= (\mathcal{M}^{\otimes 2})' =
\{S \in \Cdd \;\text{such that}\;  \\
&[S,M{\otimes}\unity_d {+} \unity_d{\otimes}M]=0 
\,\text{ for all }\,  M\in\mathcal{M} \subseteq \Cd \}. \nonumber
\end{align*}
The tensor-square commutant always contains (the subspace spanned by) the identity 
$\unity_{d^2}$
and the {\sc swap} or commutation matrix $K_{d,d}$~\footnote{ 
	The $d^2{\times}d^2$ matrix $K_{d,d}$ \cite{ZS11,HJ2,HS81} 
	permutes $A,B \in \C^{d\times d}$ in
	$K_{d,d} (A{\otimes}B) =\allowbreak (B{\otimes}A) K_{d,d}$,
	which implies $K_{d,d}\in \ts{\mathcal{M}}$ as
	$K_{d,d} (M{\otimes}\unity_d {+}\allowbreak \unity_d{\otimes}M)  
	= \allowbreak
	(M{\otimes}\unity_d {+}\allowbreak \unity_d{\otimes}M) K_{d,d}$.}. \nocite{HJ2,HS81}
Also,
the quadratic symmetries
include all linear ones, i.e.,
$S_1 {\otimes}\unity_d {+}\allowbreak \unity_d{\otimes}S_1 \in \ts{\mathcal{M}}$
for $S_1 \in\mathcal{M}'$. 
And by Jacobi's identity~\footnote{
	In this case, Jacobi's identity says that a symmetry in $\ts{\mathcal{M}}$
	commutes with 
	the commutator 
	$[M_1{\otimes}\unity_d\allowbreak {+} \unity_d{\otimes}M_1,\allowbreak\,
	M_2{\otimes}\unity_d\allowbreak {+} \unity_d{\otimes}M_2]
	=\allowbreak [M_1,M_2] {\otimes} \unity_d\allowbreak {+}  \unity_d {\otimes} [M_1,M_2]$
	if it commutes with both $M_i{\otimes}\unity_d {+} \unity_d{\otimes}M_i$ for $M_i {\in} \mathcal{M}$.},
one finds
$\ts{(\mathcal{M}_1)}=\ts{(\mathcal{M}_2)}$ if 
$\lie{\mathcal{M}_1}=\lie{\mathcal{M}_2}$.
As above, in our context this implies that
$iH_{p+1}$ cannot be simulated by $\PP$ unless
$\ts{\PP} = \ts{(\PP\cup\{iH_{p+1}\})}$ holds.

\begin{repexample}{ex_one}[completion]
The relevant tensor-square commutants have different dimensions 
$\dim[\ts{\PP}]=4$ and  $\dim[\ts{(\PP\cup\{iH_{\sf zz}\})}]=2$,
so $iH_{\sf zz}$ cannot be simulated. Naturally,  $\ts{(\PP\cup\{iH_{\sf zz}\})}$
contains 
$\unity_{16}$ 
and the commutation matrix $K_{4,4}$, which is related
to the joint permutation $(1,3)(2,4)$ of tensor components in $\C^{16\times 16}$, while 
$\ts{\PP}$ contains 
two additional quadratic  symmetries
related to the separate permutations $(1,3)$ and $(2,4)$;
see Fig.~\ref{spins}.
Evidently, the local interactions of $\PP$ 
cannot generate entanglement. Hence, a quadratic symmetry in $\ts{\PP}$ has a 
physical interpretation as an entanglement invariant. Indeed, 
the concurrence \cite{Woo:1998} of a two-qubit pure state $\ket{\psi}$ can be defined as
$[\bra{\psi}\bra{\psi}\unity_{16}\allowbreak {-} M_{(1,3)}\allowbreak {-} M_{(2,4)}\allowbreak {+} 
M_{(1,3)(2,4)} \ket{\psi}\ket{\psi}]^{1/2}/2$ 
\cite{MKB04,MB07,OS05,OS12},
where the matrix $M_p$ is defined by the permutation $p$. 
Any quadratic symmetry $S\in\ts{\PP}$ relates to a
degree-two polynomial invariant $\tr[\rho{\otimes}\rho\, S]$ 
in the entries of the density matrix $\rho$ \cite{GRB98}.
\end{repexample}

\begin{figure}[tb]
\includegraphics[]{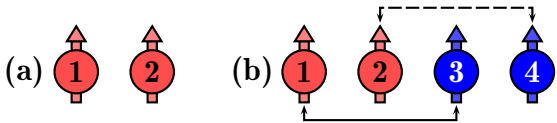}
\caption{(Color online) Visualization of Example~\ref{ex_one}. (a) No linear symmetries besides the identity
exist for both the fully controllable system and  the local interactions. (b) The 
doubled system reveals nontrivial quadratic symmetries 
corresponding to separate permutations $(1,3)$ and $(2,4)$.
\label{spins}}
\vspace{-0.4cm}
\end{figure}

Remarkably, symmetries beyond quadratic ones (i.e., those of the tensor square) are not required
for a necessary and sufficient condition for simulability~\footnote{\protect{
	Reference~\cite{ZS11} showed for controllability, i.e., $\lie{\PP}\subseteq \su(d)$, that
	$\lie{\PP}= \su(d)$ iff $\dim(\ts{\PP})=2$.
	Note $\dim[\ts{\su(d)}]=2$. 
	See \cite{ZZK14, ZZ15b} for similar results with subalgebras of $\su(d)$.}}.
Concerning the tensor-square commutant, we build on two important classification-free 
results of \cite{ZZ15b} for 
compact Lie algebras \cite{Bourb89,Bourb08b} (as
generated by skew-Hermitian matrices $iH_\nu$):
For $\lie{\PP{\cup}\QQ}$ being {\em semisimple} (and compact), Ref.~\cite{ZZ15b} first  shows
that $\lie{\PP}=\lie{\PP{\cup}\QQ}$ holds if and only if
$\dim[\ts{\PP}]=\dim[\ts{(\PP{\cup}\QQ)}]$.
Beyond the semisimple case, any compact Lie algebra $\fg$
can be uniquely decomposed as
$\fg = \fs\oplus \fc$ into its semisimple part $\fs$ and its center $\fc$
(where $\fs:=[\fg,\fg]$ and $[\fg,\fc]=0$
\cite{Note\thecentfoot}).
So Ref.~\cite{ZZ15b} secondly verifies that
the semisimple parts of $\lie{\PP}$ and $\lie{\PP{\cup}\QQ}$
have to agree if
$\dim[\ts{\PP}]=\dim[\ts{(\PP{\cup}\QQ)}]$. 
When generalizing from semisimple to arbitrary compact Lie algebras, 
the equality of the two tensor-square commutants
implies that $\lie{\PP}$ and $\lie{\PP{\cup}\QQ}$ agree--except for
the central elements (commuting with all the other ones). 
These commuting interactions require condition (B) to fix the central projection
thus resulting in the following complete characterization:

\begin{result}\label{res_one}
Consider two sets  
$\PP:=\{iH_1,\ldots,iH_p\}$ and $\QQ:=
\{iH_{p+1},\ldots,iH_q\}$ of (skew-Hermitian) interactions, and
let $C_\alpha$ denote elements of a linear basis spanning
the center $\cent$ of the commutant $(\PP{\cup}\QQ)'$.
For the central projections, define the matrix $T$ by its entries
$T_{\alpha\beta}:=
\tr(C_\alpha^\dagger iH_\beta)$ for $1 \leq \alpha \leq \dim(\cent)$ and $1  \leq \beta \leq q$
as well as
$\widetilde{T}$ by $\widetilde{T}_{\alpha\beta}:=
\tr(C_\alpha^\dagger iH_\beta)$ for $1  \leq \beta \leq p$.
Then $\PP$ simulates $\QQ$ in the sense $\lie{\PP}=\lie{\PP{\cup}\QQ}$, 
if and only if both conditions
(A)~$\dim[\ts{\PP}] = \dim[\ts{(\PP{\cup}\QQ)}]$ and
(B)~$\rank(\widetilde{T})=\rank(T)$ 
are fulfilled. 
\end{result}

Condition (B) of Result~\ref{res_one} is a basic linear-algebra
test solely based on linear symmetries.
Each of the matrices $\widetilde{T}$ and $T$ depends on both
$\PP$ and $\QQ$. 
In Example~\ref{ex_one},
$iH_{\sf zz}$ could not be generated as condition (A) is not satisfied.
Before proving Result~\ref{res_one},  the following example provides a helpful
illustration of condition (B):

\begin{example} \label{ex_two}
In a two-qubit system, consider a dipole coupling combined with a tilted magnetic field, 
i.e.,
$\PP:=\{
i(2Z_1Z_2\allowbreak {-} X_1X_2\allowbreak
{-} Y_1Y_2),\allowbreak\,
i(X_1\allowbreak {-} Y_1\allowbreak 
{+}X_2\allowbreak
{-}Y_2)
\}$.
We investigate whether
a Heisenberg-type interaction of the form
$\QQ_a:=\{i(X_1X_2\allowbreak
{+}Y_1Y_2\allowbreak
{+}Z_1Z_2)\}$
or one particular interaction of pairing 
type  (i.e., $\QQ_b:=
\{i(X_1Z_2\allowbreak
{+}Z_1X_2\allowbreak
{+}Y_1Z_2\allowbreak
{+}Z_1Y_2)\}$)
can be simulated.
Condition (A) is satisfied in both cases as 
the quadratic symmetries of $\PP$, $\PP{\cup}\QQ_a$, and 
$\PP{\cup}\QQ_b$ all coincide (there are $16$ of them).
The three linear symmetries also agree. Moreover, with the mutually commuting
operators
\begin{equation*}\label{ex2_comm}
\left(
\begin{smallmatrix}
1&0&0&0\\
0&1&0&0\\
0&0&1&0\\
0&0&0&1
\end{smallmatrix}
\right),
\left(
\begin{smallmatrix}
1&0&0&0\\
0&0&1&0\\
0&1&0&0\\
0&0&0&1
\end{smallmatrix}
\right), \;\text{and}\;
\left(
\begin{smallmatrix}
\phantom{-}0&\phantom{-}0&\phantom{-}0&1\\
\phantom{-}0&\phantom{-}0&-i&0\\
\phantom{-}0&-i&\phantom{-}0&0\\
-1&\phantom{-}0&\phantom{-}0&0
\end{smallmatrix}
\right)
\end{equation*}
forming a basis of the commutants $\PP' = (\PP{\cup}\QQ_a)'=(\PP{\cup}\QQ_b)'$,
they also span the (three-dimensional) center~$\cent$.
For the central projections, one thus gets the matrices
\begin{equation*}
T_a=
\left(
\begin{smallmatrix}
0&0&\phantom{-}0\phantom{i}\\
0&0&\phantom{-}6i\\
4&0&-4\phantom{i}
\end{smallmatrix}
\right),\,
T_b=
\left(
\begin{smallmatrix}
0&0&0\\
0&0&0\\
4&0&0
\end{smallmatrix}
\right), \;\text{and}\;
\widetilde{T}_a=\widetilde{T}_b=
\left(
\begin{smallmatrix}
0&0\\
0&0\\
4&0
\end{smallmatrix}
\right).
\end{equation*} 
Condition (B) reveals $\rank(\widetilde T_a)\neq \rank(T_a)$, $\rank(\widetilde T_b)= \rank(T_b)$,  so
$\QQ_a$ cannot be simulated by $\PP$, while $\QQ_b$ can.
Note
the isomorphy types of $\lie{\PP}$,  $\lie{\PP{\cup}\QQ_a}$, and $\lie{\PP{\cup}\QQ_b}$
are $\su(2){\oplus}\uu(1)$, $\su(2){\oplus}\uu(1){\oplus}\uu(1)$, and $\su(2){\oplus}\uu(1)$.
\end{example}

\textit{Proof of Result~\ref{res_one}}.---Decompose
the compact Lie algebras $\lie{\PP}$ and $\lie{\PP{\cup}\QQ}$
into their semisimple parts and centers. If condition (A) holds, 
the semisimple parts coincide, 
$\lie{\PP}=\fs + \fc_{\PP}$ and $ \lie{\PP{\cup}\QQ} = \fs + \fc_{\PP{\cup}\QQ}$
with $\fc_{\PP} \subseteq \fc_{\PP{\cup}\QQ}$. 
Take the
unique decomposition $iH_\ell = iH_\ell^\fs +iH_\ell^\fc$ with
$iH_\ell^\fs \in \fs$ and $iH_\ell^\fc \in \fc_{\PP{\cup}\QQ}$. Since
$[iH_n, iH_m] \in \fs$, the real-linear span of 
$\{ iH_1^\fc,\ldots,iH_p^\fc\}$ agrees with $\fc_{\PP}$, while the one of
$\{ iH_1^\fc,\ldots, iH_q^\fc\} $ equals $\fc_{\PP{\cup}\QQ}$. 
It follows that
$\fc_{\PP}=\fc_{\PP{\cup}\QQ}$  iff the dimensions of 
the two real-linear spans agree; this in turn is equivalent to the dimensions of
the two complex-linear spans being equal because all relevant Lie algebras
are compact (see Corollary~1 of Theorem~1 in Chapter~IX, Sec.~3.3 of \cite{Bourb08b}).

The center $\fc$ of a compact 
Lie algebra $\fg$ lies within the center $\cent$ of its matrix commutant $\fg'$:
$[c,g]=0$ for $c\in\fc$, $g \in \fg$ implies
$\fc \subseteq \fg'$; likewise
$[s,c]=0$ for $s \in \fg'$ shows that $\fc \subseteq \cent$.
Given the basis $\{ C_\alpha \}$ of $\cent$, introduce its dual
basis $\{ B_\alpha \}$ with respect to the Hilbert-Schmidt scalar product via
$\tr(C_\alpha^\dagger B_\beta)=\delta_{\alpha, \beta}$. So
any $C \in \cent$ can  be written  as 
$C = \sum_\alpha \tr(C_\alpha^\dagger C)\, B_\alpha$. 
Define the matrix $K$ entrywise by
$K_{\alpha\beta}:=
\tr(C_\alpha^\dagger iH_\beta^\fc)$ for $1 \leq \alpha \leq \dim(\cent)$ and $1  \leq \beta \leq q$, and similarly 
$\widetilde{K}$ by $\widetilde{K}_{\alpha\beta}:=
\tr(C_\alpha^\dagger iH_\beta^\fc)$ for $1  \leq \beta \leq p$. Hence
the dimension of $\fc_{\PP{\cup}\QQ}$ agrees with the rank of $K$, and the 
dimension of  $\fc_{\PP}$ equals the rank of $\widetilde{K}$.

Now, for any $b \in \fs$, there are two elements 
$ b_1, b_2 \in \fs$ with $b = [b_1, b_2]$
(Chapter~I, Sec.~6.4, Proposition~5  of \cite{Bourb89}).
Thus $\tr(C_\alpha b) = \tr(C_\alpha [b_1, b_2]) =\allowbreak
\tr(C_\alpha b_1 b_2)\allowbreak -\tr(C_\alpha b_2 b_1) =\allowbreak
\tr(C_\alpha b_1 b_2)\allowbreak -\tr( b_2 C_\alpha b_1) =\allowbreak
\tr(C_\alpha b_1 b_2)\allowbreak -\tr( C_\alpha b_1 b_2)=0
$, where the third equality follows as $C_\alpha \in \cent$ and $b_2$
commute, and cyclic permutations in the trace imply the fourth equality.
Moreover, $\tr(C_\alpha^\dagger b)= - \tr(C_\alpha^\dagger b^\dagger) = 
- \overline{\tr(C_\alpha b)}=0$.
Hence, $\tr(C_\alpha^\dagger iH_\beta)= \tr(C_\alpha^\dagger iH_\beta^\fc)$, which implies $T= K$
and $\widetilde{T}= \widetilde{K}$. In summary, 
$\fc_{\PP}=\fc_{\PP{\cup}\QQ}$  iff the ranks of
$T$ and $\widetilde{T}$ agree,
which proves Result~\ref{res_one}. $\hfill\square$

\section{Algorithmics and beyond}
Both linear and quadratic symmetries 
can readily be computed by standard linear algebra:
Linear symmetries $S\in \Cd$ are determined by the commutant
and can be obtained by solving the linear equations
$(\unity_d {\otimes} M {-} M^{t} {\otimes} \unity_d) \opvec(S) = 0$
jointly for all 
$M\in \mathcal{M} \subseteq \Cd$  \cite{ZS11,HS81}. Here,
$\opvec(S)$ is a column vector of length $d^2$ stacking
all columns of $S$ \cite{HJ2}. The dimension of the solution  
is 
$d^2{-}r$ where $r$ denotes the rank of the matrix formed by vertically stacking
the matrices $\unity_d {\otimes} M {-} M^{t} {\otimes} \unity_d$.
Likewise, the quadratic symmetries $S\in \C^{d^2\times d^2}$ 
(given by the tensor-square commutant) 
just amount to solving 
$[\unity_{d^2} {\otimes} (M{\otimes}\unity_d\allowbreak {+} \unity_d{\otimes}M)  \allowbreak {-} 
(M{\otimes}\unity_d\allowbreak {+} \unity_d{\otimes}M)^{t} {\otimes} \unity_{d^2}]
\opvec(S) = 0$ jointly for all $M\in \mathcal{M}$. 
The preceding discussion explains how to explicitly determine 
linear and quadratic symmetries. This allows us to test condition~(A) (i.e.,  
$\dim[\ts{\PP}] = \dim[\ts{(\PP{\cup}\QQ)}]$) by comparing
the dimensions of the quadratic symmetries 
for $\PP$ and $\PP{\cup}\QQ$.

As the commutant $(\PP{\cup}\QQ)'$
represents the linear symmetries of $\PP{\cup}\QQ$, its center $\mathcal{C}$ is readily obtained
by solving the linear equations $(\unity_d {\otimes} M {-} M^{t} {\otimes} \unity_d) \opvec(C) = 0$
and $(\unity_d {\otimes} S {-} S^{t} {\otimes} \unity_d) \opvec(C) = 0$ jointly 
with $M$ extending over all $M \in \PP{\cup}\QQ$ and $S$ over all 
$S \in (\PP{\cup}\QQ)'$. Solving for $C$ yields a  basis $C_\alpha$
of the center $\mathcal{C}$, and one can determine
the matrices $T$ and $\tilde{T}$ as
$T_{\alpha\beta}=\tr(C_\alpha^\dagger iH_\beta)$
for $1 \leq \alpha \leq \dim(\cent)$ and $1  \leq \beta \leq q$,
as well as
$\widetilde{T}_{\alpha\beta}=
\tr(C_\alpha^\dagger iH_\beta)$ for $1  \leq \beta \leq p$. Since condition~(B) is given by 
$\rank(\widetilde{T})=\rank(T)$, it can easily be tested by elementary linear-algebra computations 
comparing the ranks of $\tilde{T}$ and $T$. 
To sum up, Result~\ref{res_one} reduces the Hamiltonian membership problem to 
straightforward solutions of homogeneous linear equations.

\begin{example}[central-spin model]\label{ex_central}
Consider a central spin 
interacting with $n{-}1$ surrounding spins via a star-shaped coupling graph
(where the surrounding spins may be taken as uncontrolled spin bath) \cite{Gaudin76,BS07,AGB14}.
The interactions amount to  a drift term (tunneling plus coupling) and just 
a local  $Z$-control on the central spin, 
$\PP:=\{
i X_1\allowbreak 
{+} i\sum_{k=2}^{n} J_k (X_1X_k\allowbreak
 {+} Y_1Y_k\allowbreak {+} Z_1Z_k),\allowbreak\,
i Z_1\}$.  We ask whether 
the central spin can be fully controlled,
i.e.,
if $\QQ:=\{i X_1\}$ can be simulated.
Depending on the interaction strengths $J_k \in \R$ for $k\geq 2$,
different cases 
are possible: 
(a) with $J_k=1$ and  (b) with $J_k=2$ for even $k$, and $J_k=1$ otherwise.
\end{example}
Computational results for the central spin model
have been obtained using exact arithmetic \cite{Magma}
for a moderate number of spins
as detailed  in Table~\ref{tab_central}.
These results vary significantly for different coupling strengths $J_k$. But our approach 
for deciding simulability allows for analytic reasoning
even beyond specific choices of $J_k$. 
For Hamiltonian
simulation, it thus provides a powerful technique to analyze and understand the dynamics 
of general quantum systems. This even holds  
if the symmetries cannot be calculated explicitly.
Showcases for 
the strength of explicit symmetries are given in Examples~\ref{ex_one}-\ref{ex_central}, while
Example~\ref{ex_central} also makes use of symmetries implicitly (in parts where they cannot be calculated explicitly)
via the proofs of  the Appendix~\ref{App:proofs}.
These proofs motivate the following:

\begin{table}[t]
\caption{Central-spin model of Example~\ref{ex_central}:
number $n$ of spins, 
Lie dimensions $\dim(\lie{\PP})=\dim(\lie{\PP{\cup}\QQ})$,
the isomorphy type, 
dimensions of quadratic and linear
symmetries (i.e.\
$\dim[\ts{\PP}]=\dim[\ts{(\PP{\cup}\QQ)}]$ and $\dim[\PP']=\dim[(\PP{\cup}\QQ)']$),
and ranks of the central projections [i.e., $\rank(\widetilde{T})=\rank(T)$].\label{tab_central}}
\begin{tabular}[t]
{@{\hspace{2mm}}l@{\hspace{2mm}}r@{\hspace{2mm}}l@{\hspace{3mm}}r@{\hspace{-19mm}}r@{\hspace{2mm}}r@{\hspace{1mm}}}\\[-3mm]
\hline\hline
 $n$
 &Lie-\,&  Isomorphy & &  
 No.\  of symmetries
& Rank of\\[-1mm]
& dim. &  type & quad.\ & lin.\  & c.\ proj.\
\\[0.5mm]
\hline\\[-4mm]
\multicolumn{5}{l}{case (a): $J_k=1$}\\
     2 & 15 & $\su(4)$ & 2 & 1 & 0\\
     3 & 38 & $\su(2){\oplus}\su(6)$ & 8 & 2 & 0\\
     4 & 78 & $\su(4){\oplus}\su(8)$ & 50 & 5 & 0\\
     5 & 137 & $\su(2){\oplus}\su(6){\oplus}\su(10)$ & 392 & 14 & 0\\
     6 & 221 & $\su(4){\oplus}\su(8){\oplus}\su(12)$ & 3528 & 42 & 0 \\[0.5mm]
\hline\\[-4mm]
\multicolumn{5}{l}{case (b): $J_k=2$ for even $k$ and $J_k=1$ otherwise}\\
     2 & 15 & $\su(4)$ & 2 & 1 & 0\\
     3 & 63 & $\su(8)$ & 2 & 1 & 0\\
     4 & 158 & $\su(4){\oplus}\su(12)$ & 8 & 2 & 0\\
     5 & 396 & $\su(2){\oplus}\su(6){\oplus}\su(6){\oplus}\su(18)$ & 32 & 4 & 0\\
     6 & 796 & $\su(4){\oplus}\su(8){\oplus}\su(12){\oplus}\su(24)$ & 200 & 10 & 0 \\[0.5mm]
\hline
\hline
\end{tabular}
\end{table}

\begin{conjecture*}\label{res_central}
In the central-spin model of Example~\ref{ex_central},
the central spin is fully controllable for a finite number of spins
and any choice of $J_k$ (i.e., 
$iX_1$ can be simulated, and 
the surrounding spins can be uncoupled by applying the control).
\end{conjecture*}

\section{Discussion}
Similar to the Hamiltonian membership problem 
for interactions
solved here, one may address membership for groups,
e.g., (i) in the (prototypical) discrete case, 
(ii) in connected compact Lie groups, and (iii) in non-connected compact groups
including finite groups. 

In {\em discrete groups} (i), asking the question (a) if $\hat\QQ=\{U_{p{+}1}\}$ is (exactly)
contained in the group generated by the unitaries 
$\hat\PP=\{U_1,\ldots,U_p\}$ 
is undecidable for $\SU(N)$ (at least for $N\geq 4$) \cite{Jea:2005}.
Yet
the question (b) of approximate universality 
\cite{Sho:1996,Kit02}, i.e., if all unitaries in $\SU(N)$ can be approximated, 
is decidable 
\cite{Jea:2004,Jea:2005} by comparing the matrix algebra generated by elements 
$\bar{U}_{\nu}{\otimes} U_{\nu}$ for  $U_\nu\in\hat\PP$ with its equivalent for $\SU(N)$ (plus other conditions).
Still, the tedious algebra closure is needed, similar to the Lie closure.
Question (b) is equivalent to comparing the {\em topological closure}
of the group generated by $\hat\PP$ to $\SU(N)$ and thus leads to (ii).

In {\em continuous groups} (ii), Result~\ref{res_one} applies to decide
if two connected, compact Lie groups (given by their 
infinitesimal generators) are equal:

\begin{result}
Given two sets  
$\PP$ and  $\QQ$ of (skew-Hermitian) interactions, the elements of $\PP$ simulate the ones
of $\QQ$ \emph{and vice versa} iff both $\lie{\PP}=\lie{\PP{\cup}\QQ}$ and
$\lie{\QQ}=\lie{\PP{\cup}\QQ}$ hold, where each condition can be tested by
Result~\ref{res_one}.
\end{result}
Our findings do not generalize to {\em non-connected compact groups} (iii),
nor are they implied by the representation theory
of compact groups. In particular,
finite groups with trivial quadratic symmetries $[S,U_\nu {\otimes}U_\nu]=0$ only
(known as group designs \cite{GAE07}) do not contradict our work.

\section{Conclusion}
We have presented a complete symmetry approach to
decide Hamiltonian simulability, i.e., whether given drift and control
Hamiltonians can simulate a target (effective) Hamiltonian in finite dimensions.
Quadratic symmetries lead to an understanding that allows one to 
algebraically prove simulability in classes of
many-body systems where
the usual computational assessment via the Lie closure
is infeasible. This is exemplified
by proving simulability for interesting cases of
the central-spin model (see the Appendix~\ref{App:proofs})
for which only very restricted cases were addressed before \cite{AGB14}.

Achievability of specific target interactions
is particularly important for fault-tolerance,
where the simulation of a particular Hamiltonian
(or universality) is needed only on logical subspaces and not globally.
While \emph{linear} symmetries have often been used in those
cases \cite{ZanRas97,Zanardi99,VKL99a,WuLi02},
going a step further by applying quadratic symmetries to ensure controllability
or simulability on a noise-protected subspace could be an interesting
application, simplifying complicated system-algebraic analysis. 
For instance, in Ref.~\cite{PRL_decoh}, we examined standard scenarios
of noise-protected subspaces, where controllability was (moderately) easy to assess.
However, in more realistic settings, analyzing quadratic symmetries
and their restrictions to protected subspaces is anticipated to be much
easier than establishing Lie closures over restricted subspaces.

Moreover, our results on \emph{quadratic} symmetries distinguishing 
local properties from global ones can be generalized into an overarching framework
that encapsulates concurrence (Example~\ref{ex_one}) and links naturally to
entanglement detection via a \emph{quadratic} invariant of the quantum system under
local transformations in \cite{KB09,KKK10,OK12,OK13}.

Our findings imply that for \emph{any} nonsimulable interaction, a related resource is lacking. 
In Example 1 it simply was entanglement, but more generally we can characterize lacking resources  as
induced by conserved quantites arising from quadratic symmetries. 
This paves the way toward a resource theory of quantum simulability.

\begin{acknowledgments}
Z.Z.\ acknowledges funding by
the British {\em Engineering and Physical Sciences Research Council} ({\sc epsrc}).
R.Z.\ and T.S.H.\ are supported by {\em Deutsche Forschungsgemeinschaft}
({\sc dfg}) in the collaborative research center {\sc sfb} 631 and via
Grants No. {\sc gl} 203/7-1 and No. 203/7-2, and by the
\emph{EU programmes} {\sc siqs} and {\sc quaint} and the Bavarian network of excellence {\sc e}x{\sc qm}.
\end{acknowledgments}

\appendix

\allowdisplaybreaks

\section{Central-spin controllability for different levels of generality\label{App:proofs}}

In this Appendix, we analyze under which conditions on the coupling coefficients $J_k$ 
the central spin is controllable in Example~\ref{ex_central}. We collect proofs 
for this controllability under varying assumptions.
Recall the set  $\PP=\{iH_1,\, iH_2\}$ of control interactions, where $iH_1=i X_1\allowbreak 
+ i\sum_{k=2}^{n} J_k (X_1X_k\allowbreak
 {+} Y_1Y_k\allowbreak {+} Z_1Z_k)$ and $iH_2=i Z_1$, as well as the 
 target interaction $\QQ=\{iX_1\}$.
Assuming that condition (A) of Result~\ref{res_one} holds, there exists an element $iC$ in the center of 
$\lie{\PP {\cup} \{ iX_1\}}$ such that $iX_1 {+} iC \in \lie{\PP}$. Since
$iZ_1 \in \lie{\PP}$, one obtains $-[iZ_1,[iZ_1, iX_1{+}iC]]/4= iX_1 \in \lie{\PP}$.
Thus, it follows that it suffices to verify condition (A) in the different cases below, which is 
equivalent to showing that   
$D(iX_1) v =0$ holds for all vectors $v \in \C^{d^4}$ with $d:=2^n$
and $D(iH_1) v = D(iH_2)v =0$. Here, the linear operator
\begin{align*}
D(M):=&[\unity_{d^2} {\otimes} 
(M{\otimes}\unity_d{+} \unity_d{\otimes}M)   \\
&{-} 
(M{\otimes}\unity_d{+} \unity_d{\otimes}M)^{t} {\otimes} \unity_{d^2}]
\in \C^{d^4 \times d^4}
\end{align*}
is a shortcut in order to define the  linear equations $D(M)v=0$ for 
the matrix $M \in \Cd$ and
quadratic symmetries $S$ where
$v:=\opvec(S)$. One naturally obtains that both of the equations 
$[D(M_1), D(M_2)]\allowbreak = D([M_1,M_2])$ and 
$\exp[D(M_1)] D(M_2) \exp[-D(M_1)]\allowbreak = D[\exp(M_1) M_2 \exp(-M_1)]$
hold for all matrices $M_1, M_2$.

\begin{proposition}\label{prop:1}
The interaction $iX_1$ can be simulated if all couplings $J_k$ are either
(a) equal, i.e., $J_k = J$,
(b) equal up to an odd integer $o_k$, i.e., $J_k = J o_k$ where $o_k$ may depend on $k$, or
(c) $\Q$-linear independent.
\end{proposition}

\begin{proof}
Consider the definitions
$i\tilde{H}_{\sf zz} := i\sum_{k=2}^{n} J_k Z_1Z_k\allowbreak  = iH_1\allowbreak + [iH_2,[iH_2, iH_1]]/4
\in \lie{\PP}$ and $i\tilde{H}:=i X_1\allowbreak 
+ i\sum_{k=2}^{n} J_k (X_1X_k\allowbreak
 {+} Y_1Y_k)= iH_1-i\tilde{H}_{\sf zz} \in \lie{\PP}$.
 We
assume in the following that $D(iH_1) v = D(iH_2)v=D(i\tilde{H}_{\sf zz}) v = D(i\tilde{H})v =0$ holds
in order to prove $D(iX_1) v =0$.

The joint eigenbasis of the operators $D(iZ_k/2)$ and  $D(iZ_1 Z_k/2)$ for $k\in\{2, \ldots, n\}$
is given by the computational basis, and its basis vectors are 
 $w(b)=\ket{b_1}  \otimes \cdots\otimes \ket{b_{4n}}$ with $b_k \in\{0,1\}$,
$\ket{0}:=(0,1)^t$, and $\ket{1}:=(1,0)^t$.
This implies that the eigenvalue equations are $D(iZ_k/2)w(b)=i\mu_k(b) w(b)$ and 
$D(iZ_1Z_k/2)w(b)=i\lambda_k(b) w(b)$, and
the corresponding
eigenvalues are given by 
\begin{align*}
\mu_k(b)=\, & \tfrac{1}{2}({-}s_k\allowbreak{-}s_{n+k}\allowbreak{+}s_{2n+k}\allowbreak
{+}s_{3n+k}), \\
\lambda_k(b)=\, & \tfrac{1}{2}({-}s_1s_k\allowbreak {-} s_{n+1}s_{n+k}\allowbreak{+}
s_{2n+1}s_{2n+k}\allowbreak
{+}s_{3n+1}s_{3n+k}), \allowbreak 
\end{align*}
where $\mu_k(b)\in\{-2,-1,0,1,2\}$, $\lambda_k(b)\in\{-2,-1,0,1,2\}$, and
$s_j:=2b_j{-}1$.  By checking all of the $2^8$ cases for $s_{en+1},s_{en+k} \in\{{-}1,{+1}\}$ 
and $e\in \{ 0,1,2,3\}$, one concludes that ${\mu_k(b) \bmod 2} = \lambda_k(b) \bmod 2$ holds if
$D(iZ_1)w(b)=0$.
Recall that $D(iZ_1)v=D(i\tilde{H}_{zz})v=0$
and expand $v$ as 
$v=\sum_{b} \alpha_{b} w(b)$.
It follows that the equations $D(iZ_1)w(b)=0$ and $ D(i \tilde{H}_{zz}) w(b) =0 $ hold for $\alpha_{b}\neq 0$
as each $w(b)$ is an eigenvector of $D(iZ_1)$ and
$D(i\tilde{H}_{zz})=\sum_{k=2}^{n} 2 J_k \, D(iZ_1Z_k/2)$.
Assuming $D(iZ_1)w(b)=0$, this also means that 
the relation $\mu_{z}(b) \bmod 2 = \lambda_{zz}(b) \bmod 2$ holds
for the eigenvalue $i\mu_{z}(b)$ of $\sum_{k=2}^nD(i Z_k/2)$ and 
the  eigenvalue $i\lambda_{zz}(b)$ of $\sum_{k=2}^{n} D(iZ_1Z_k/2)$.
Moreover, we obtain for $\alpha_{b}\neq 0$ that $0= D(i \tilde{H}_{zz}) w(b)
= \allowbreak [i\sum_{k=2}^n 2 J_k \lambda_k(b)] w(b)$ and, consequently,
$\sum_{k=2}^n 2 J_k \lambda_k(b)=0 $.

The proof depends now on the particular cases, and we prove in each case that
$\mu_z(b) \bmod 2 =0$: 
For the case (a) with $J_k=J$, it follows that  $\lambda_{zz}(b) = \sum_{k=2}^n  \lambda_k(b)
=0$. This implies that $\mu_z(b) \bmod 2 =0$. 
In case (b), we obtain $J_k=J o_k$
and $\lambda_{zz}(b) \bmod 2 = \sum_{k=2}^n  \lambda_k(b) \bmod 2 = \sum_{k=2}^n o_k  \lambda_k(b) \bmod 2 = 0$,
which also shows that $\mu_z(b) \bmod 2 =0$. 
For case (c), $\sum_{k=2}^n 2 J_k \lambda_k(b)=0 $
means that $\lambda_k(b)=0$ for all $k$ since  the couplings $J_k$ are $\Q$-linear independent 
and $\lambda_k(b)\in\Z$. In particular, it follows that  $\lambda_{zz}(b) = \sum_{k=2}^n  \lambda_k(b) =0$,
which proves again that $\mu_z(b) \bmod 2 =0$.

Define the operator
\begin{align*}
W&:=\exp[\pi\sum_{k =2}^n \, D(iZ_k/2)].
\end{align*}
Using the properties of $\sum_{k =2}^n \, D(iZ_k/2)$,
one gets that   the equation 
$W w(b)= e^{i \mu_z(b) \pi} w(b)=w(b)$ holds for each  $w(b)$  with $\alpha_{b}\neq 0$, 
where the last equality follows from $\mu_z(b) \bmod 2=0$ .
Thus, we obtain $Wv=W \sum_{b} \alpha_b w(b)= \sum_{b} \alpha_b Ww(b)= \sum_{b} \alpha_b w(b)=v$.
We also have that $W iD(X_1 X_k {+}Y_1Y_k) W^{\dagger}
=\allowbreak i D[G \allowbreak (X_1 X_k\allowbreak{+}Y_1Y_k)\allowbreak G^{\dagger}]$,
using the notation
\begin{equation*}
G:=\exp(\pi\sum_{k=2}^n \, iZ_{k}/2)
\allowbreak= \prod_{k=2}^n \exp(\pi\, iZ_{k}/2) 
= \prod_{k=2}^n 
 i  Z_{k}.
 \end{equation*} 
It follows that
 $W iD(X_1 X_k {+}Y_1Y_k) W^{\dagger}= iD[\prod_{k', k''=2}^n\allowbreak 
 (i Z_{k'}) (X_1 X_k {+}Y_1Y_k) ({-}iZ_{k''})] = - iD(X_1 X_k {+}Y_1Y_k) $ since
 $Z_kX_k Z_k=-X_k$ and $Z_kY_k Z_k=-Y_k$.
Naturally, $WD(iX_1)W^{\dagger}=D(iX_1)$ is also satisfied.

One can now verify that 
\begin{align*}
0&=W D(i\tilde{H})v\allowbreak= 
W D(i\tilde{H})W^\dagger Wv\allowbreak\\
&=
 [D(iX_1) \allowbreak - \sum_{k=2}^n iJ_k D(X_1 X_k {+}Y_1Y_k)]\, v\allowbreak\\
&=
D(i\tilde{H})v \allowbreak - 2 \sum_{k=2}^n iJ_k D(X_1X_k \allowbreak {+} Y_1 Y_k) v.
\end{align*} 
This implies
$\sum_{k=2}^n iJ_k D(X_1X_k \allowbreak {+} Y_1 Y_k)v=0$,
and one concludes that $D(i\tilde{H}) v- \sum_{k=2}^n
iJ_k D(X_1X_k \allowbreak {+} Y_1 Y_k)v= D(iX_1)v=0$.
\end{proof}

The techniques in the proof of Proposition~\ref{prop:1} can be generalized in
order to establish the following result:

\begin{proposition}
The interaction $iX_1$ can be simulated if $J_k=J$ for $ 2\le k \le n_0$ and 
$J_k=2 J$ for $ n_0 < k \le n$.
\end{proposition}

\begin{proof}
We establish again all the properties of the first two paragraphs in the proof of Proposition~\ref{prop:1}.
Then, it follows that $ \sum_{k=2}^{n_0} J \lambda_k(b)+ \sum_{k=n_0+1}^n 2 J \lambda_k(b) =0$.
Let $i\mu^{(0)}_z(b)$ be the eigenvalue of $\sum_{k=2}^{n_0} D(i Z_k/2)$. One obtains 
that $\mu^{(0)}_z(b) \bmod 2 =0$ for each $w(b)$ with $\alpha_{b}\neq 0$.
Define the operator
\begin{align*}
W^{(0)}&:=\exp[\pi\sum_{k =2}^{n_0} \, D(iZ_k/2)].
\end{align*}
We apply the properties of $\sum_{k =2}^{n_0} \, D(iZ_k/2)$ and
conclude that the equation  
$W^{(0)} w(b)= e^{i \mu^{(0)}_z(b) \pi} w(b)=w(b)$ holds for each element $w(b)$  satisfying $\alpha_{b}\neq 0$, 
where the last equality follows from $\mu^{(0)}_z(b) \bmod 2=0$.
Thus, we obtain $W^{(0)}v=W^{(0)} \sum_{b} \alpha_b w(b)= \sum_{b} \alpha_b W^{(0)}w(b)= \sum_{b} \alpha_b w(b)=v$.
We also have that $W^{(0)} iD(X_1 X_k {+}Y_1Y_k) (W^{(0)})^{\dagger}
=\allowbreak i D[G^{(0)} \allowbreak (X_1 X_k\allowbreak{+}Y_1Y_k)\allowbreak (G^{(0)})^{\dagger}]$
using the notation
\begin{equation*}
G^{(0)}:=\exp(\pi\sum_{k=2}^{n_0} \, iZ_{k}/2)
\allowbreak= \prod_{k=2}^{n_0} \exp(\pi\, iZ_{k}/2) = \prod_{k=2}^{n_0} 
i  Z_{k}.
\end{equation*} 
It follows that
\begin{align*}
&W^{(0)} iD(X_1 X_k {+}Y_1Y_k) (W^{(0)})^{\dagger}\\
&= iD[\prod_{k', k''=2}^{n_0}\allowbreak 
(i Z_{k'}) (X_1 X_k {+}Y_1Y_k) ({-}iZ_{k''})] \\
&= - iD(X_1 X_k {+}Y_1Y_k) 
\end{align*}
if $ 2 \le k \le n_0$
since
$Z_kX_k Z_k=-X_k$ and $Z_kY_k Z_k=-Y_k$, and  
$W^{(0)} iD(X_1 X_k {+}Y_1Y_k) (W^{(0)})^{\dagger}=iD(X_1 X_k {+}Y_1Y_k) $ if $k > n_0$.
Naturally, $W^{(0)}D(iX_1)(W^{(0)})^{\dagger}=D(iX_1)$ is also satisfied.

One can now verify that 
\begin{align*}
0&=W^{(0)} D(i\tilde{H})v\allowbreak= 
W^{(0)} D(i\tilde{H})(W^{(0)})^\dagger W^{(0)}v\allowbreak\\
&= 
[D(iX_1) - \sum_{k=2}^{n_0} i J D(X_1 X_k {+}Y_1Y_k) \\
&\phantom{=[\;} +\sum_{k=n_0 +1}^{n} 2iJ D(X_1 X_k {+}Y_1Y_k)]\, v \allowbreak\\
&=
D(i\tilde{H})v \allowbreak - 2 \sum_{k=2}^{n_0} i JD(X_1X_k \allowbreak {+} Y_1 Y_k) v,
\end{align*}
which implies
$\sum_{k=2}^{n_0} iJ_k D(X_1X_k \allowbreak {+} Y_1 Y_k)v=0$.
Thus, one can conclude that $D(i\tilde{H}) v- \sum_{k=2}^{n_0}
iD(X_1X_k \allowbreak {+} Y_1 Y_k)v= D(i \tilde{H}^{(1)})v = 0$, where we introduced the
notation $i\tilde{H}^{(1)}:=iX_1 + i \sum_{k=n_0+1}^n 2J(X_1 X_k {+}Y_1Y_k)$. 

Furthermore, the equation $D(i\tilde{H}_{zz})v= D(i\tilde{H}^{(1)})v\allowbreak=D(iH_2)v=0$ implies
the important commutator identity $[[[D(i\tilde{H}^{(1)}), D(iH_2)], D(i\tilde{H}^{(1)})], D(i\tilde{H}_{zz})]v=0$.
In addition, we have 
\begin{align*}
&[[[D(i\tilde{H}^{(1)}), D(iH_2)], D(i\tilde{H}^{(1)})], D(i\tilde{H}_{zz})]=\\
&D([[[\tilde{H}^{(1)}, H_2], \tilde{H}^{(1)}], \tilde{H}_{zz}])= 64 J^2 D(i \sum_{k=n_0+1}^n Y_k).
\end{align*}
Thus, $D(i \sum_{k=n_0+1}^n Y_k)v=0$. Now, we also get that 
\begin{align*}
0&=[[D(i\tilde{H}_{zz}), D(i \sum_{k=n_0+1}^n Y_k)], D(i \sum_{k=n_0+1}^n Y_k)]v\\
&=
-4 D(i\tilde{H}_{zz}^{(1)})v, 
\end{align*}
where $i\tilde{H}_{zz}^{(1)}:=i \sum_{k=n_0+1}^{n} 2JZ_1Z_k$.
Considering the expansion $v= \sum_{b} \alpha_b w(b)$, we obtain from $D(i\tilde{H}_{zz}^{(1)})v=0$ 
that  the condition 
$\lambda^{(1)}_{zz}(b)=0$ holds for
the eigenvalue $i\lambda^{(1)}_{zz}(b)$ of 
$\sum_{k=n_0+1}^{n}  D(i Z_1Z_k/2)$
with respect to a vector $w(b)$ with $\alpha_b \ne 0$.  
As we have for any $w(b)$ with $D(iZ_1)w(b)=0$ that the eigenvalue
$i\mu^{(1)}_z(b)$ of $\sum_{k=n_0+1}^{n} D(i Z_k/2)$ satisfies
$\lambda^{(1)}_{zz}(b) \bmod 2 = \mu^{(1)}_z(b) \bmod 2$, thus we can conclude that 
 $\mu^{(1)}_z(b) \bmod 2 =0$ for $b$ with $\alpha_b \ne 0$.
We can now define the operator
\begin{align*}
W^{(1)}&:=\exp[\pi\sum_{k =n_0+1}^{n} \, D(iZ_k/2)].
\end{align*}
Using the properties of $\sum_{k =n_0+1}^{n} \, D(iZ_k/2)$,
one gets that   the equation 
$W^{(1)} w(b)= e^{i \mu^{(1)}_z(b) \pi} w(b)=w(b)$ holds for each  $w(b)$  with $\alpha_{b}\neq 0$, 
where the last equality follows from $\mu^{(1)}_z(b) \bmod 2=0$.
Thus, we obtain $W^{(1)}v=W^{(1)} \sum_{b} \alpha_b w(b)= \sum_{b} \alpha_b W^{(1)}w(b)= \sum_{b} \alpha_b w(b)=v$.
We also have that $W^{(1)} iD(X_1 X_k {+}Y_1Y_k) (W^{(1)})^{\dagger}
=\allowbreak i D[G^{(1)} \allowbreak (X_1 X_k\allowbreak{+}Y_1Y_k)\allowbreak (G^{(1)})^{\dagger}]$
using the notation
\begin{align*}
G^{(1)}&:=\exp(\pi\sum_{k=n_0+1}^{n} \, iZ_{k}/2)
\allowbreak= \prod_{k=n_0+1}^{n} \exp(\pi\, iZ_{k}/2)\\
&\phantom{:}= \prod_{k=n_0+1}^{n} 
 i  Z_{k}.
 \end{align*} 

With these preparations, one can now verify that 
\begin{align*}
0&=W^{(1)} D(i\tilde{H}^{(1)})v\allowbreak= 
W^{(1)} D(i\tilde{H})(W^{(1)})^\dagger W^{(1)}v\allowbreak\\
&= [D(iX_1) - 2\sum_{k=n_0+1}^{n} i J D(X_1 X_k {+}Y_1Y_k) ]v \allowbreak\\
&=
D(i\tilde{H})v \allowbreak - 4 \sum_{k=n_0+1}^{n} i JD(X_1X_k \allowbreak {+} Y_1 Y_k) v,
\end{align*} 
which hence  implies
$\sum_{k=n_0+1}^{n} iJ_k D(X_1X_k \allowbreak {+} Y_1 Y_k)v=0$.
Consequently, 
one can finally conclude that $D(i\tilde{H}^{(1)}) v- \sum_{k=n_0+1}^{n}
2iJD(X_1X_k \allowbreak {+} Y_1 Y_k)v= D(i X_1)v = 0$. 
\end{proof}

\bibliographystyle{apsrev4-1}


%

\end{document}